\documentclass{article}

\usepackage{amssymb,amsfonts,amsthm}
\usepackage{graphicx}
\usepackage{subfigure}
\usepackage{amssymb}
\usepackage{amsthm}
\usepackage{amsmath}
\usepackage[font=small,labelfont=bf,tableposition=top]{caption}

\DeclareCaptionLabelFormat{andtable}{#1~#2  \&  \tablename~\thetable}

\theoremstyle{plain}
\newtheorem{theorem}{Theorem}
\newtheorem{corollary}[theorem]{Corollary}

\newtheorem{lemma}{Lemma}
\newtheorem{proposition}{Proposition}
\newtheorem*{definition}{Definition}

\theoremstyle{remark}

\def\pmb#1{\setbox0=\hbox{$#1$}%
  \kern-.025em\copy0\kern-\wd0
  \kern.05em\copy0\kern-\wd0
  \kern-.025em\raise.0433em\box0}
\def\pmbs#1{\setbox0=\hbox{$\scriptstyle #1$}%
  \kern-.0175em\copy0\kern-\wd0
  \kern.035em\copy0\kern-\wd0
  \kern-.0175em\raise.0303em\box0}

\begin{document}

\title{\bf Spherically symmetric steady states of John elastic bodies in general relativity}

\author{H\aa kan Andr\'easson\footnote{E-Mail: hand@chalmers.se}\quad \& Simone Calogero\footnote{E-Mail: calogero@chalmers.se} \\[0.2cm]
Mathematical Sciences, University of Gothenburg \\ \&
Chalmers University of Technology \\
41296 G\"oteborg, Sweden
}
\date{}
\maketitle
\begin{abstract}
We study some properties of static spherically symmetric elastic bodies in general relativity using both analytical and numerical tools. The materials considered belong to the class of John elastic materials and reduce to perfect fluids when the rigidity parameter is set to zero. We find numerical support that such elastic bodies exist with different possible shapes (balls, single shells and multiple shells) and that their gravitational redshift can be very large ($z\approx 2.8$) without violating the dominant energy condition. Moreover we show that the elastic body has finite radius even in the case when the constitutive equation of the elastic material is a perturbation of a polytropic fluid without finite radius, thereby concluding that such fluids are structurally unstable within the larger class of elastic matter models under study.
\end{abstract}
\section{Introduction}
Some astrophysical systems, such as  slowly rotating stars, are well approximated by static spherically symmetric configurations of matter. The physical properties of such systems, e.g., their mass, radius and redshift, are carefully studied in the mathematical and physical literature. The results obtained depend on the choice of the matter model being used. Popular choices are the perfect fluid model, which applies to stars~\cite{kipp, ngu1, ngu2}, and the Vlasov model of collisionless particles, which applies to galaxies~\cite{gd}.
When a perfect fluid model is used, static matter distributions are necessarily spherically symmetric~\cite{licht,lindb}. For Vlasov particles this is not longer the case but most studies on Vlasov matter are none the less concerned with spherically symmetric models (see \cite{AKR1,AKR2, Rein} for studies on axially symmetric configurations of Vlasov matter).

The purpose of this paper is to investigate some physical properties of static spherically symmetric elastic bodies in general relativity and compare them with the analogous properties of perfect fluids and Vlasov particles. An important difference between elastic matter and other matter models is that elastic bodies possess generically a solid surface, i.e., a boundary with positive matter density. As a consequence elastic bodies behave quite differently from perfect fluids even when the constitutive equation of the elastic material is an arbitrary small perturbation of the perfect fluid model. This sort of ``structural instability" of perfect fluids was already found in the context of spatially homogeneous models~\cite{CH}. Earlier studies on static spherically symmetric elastic bodies in general relativity can be found in~\cite{BCV,KS,M2,M1, park}  (see~\cite{ABS1,ABS2, ABS3, BS2} for results on weakly deformed bodies without spherical symmetry; in the Newtonian case, static configurations without symmetry assumptions and for arbitrary deformations have been constructed in~\cite{TC}).

In the rest of the Introduction we first present the general problem of static spherically symmetric matter distributions in general relativity and then we specialize to the case of elastic bodies. A more detailed discussion on the general relativistic theory of elasticity is provided in Appendix. For more background on the subject we refer to~\cite{BS, CQ, KM}.

Let $(t,r,\theta,\phi)$ be a system of Schwarzschild coordinates in a static spherically symmetric spacetime. Consider a matter distribution with energy density $\rho(r)$, radial pressure $p_\mathrm{rad}(r)$ and tangential pressure $p_\mathrm{tan}(r)$. The metric of spacetime is then given by
\[
ds^2=-e^{2\beta(r)}\,dt^2+e^{2\lambda(r)}\,dr^2+r^2(d\theta^2+\sin^2\theta d\phi^2),
\]
where, assuming regularity at the center and asymptotic flatness at infinity,
\begin{equation}\label{metricST}
e^{-2\lambda(r)}=1-\frac{2m(r)}{r}, \quad \beta(r)=-\int_r^\infty e^{2\lambda(s)}\left(\frac{m(s)}{s^2}+\frac{s}{2}\, p_\mathrm{rad}(s)\right)ds,
\end{equation}
with
\[
 m(r)=\frac{1}{2}\int_0^rs^2\rho(s)\,ds
\]
denoting the Hawking mass.
We use physical units such that $8\pi G=c=1$.
The radial pressure satisfies the TOV equation
\begin{equation}
p_\mathrm{rad}'=-(p_\mathrm{rad}+\rho)\frac{m+r^3p_\mathrm{rad}/2}{r(r-2m)}+\frac{2}{r}(p_\mathrm{tan}-p_\mathrm{rad})\label{TOVeq},
\end{equation}
where a prime denote differentiation in the radial variable.
The matter distribution is said to have a {\it boundary} at $r=R$ if $p_\mathrm{rad}(R)=0$, which is called {\it solid} if $\rho(R)>0$. 

To obtain a complete system on the matter fields we need to add additional equations on the energy density and the tangential pressure. For a barotropic fluid one postulates the existence of a function $g :[0,\infty)\to [0,\infty)$ such that $g'\geq 0$, $g(0)=0$ and $p_\mathrm{tan}=p_\mathrm{rad}=p=g(\rho)$. For Vlasov particles the matter fields are given by suitable integrals in the momentum variable of the particle density $f$ in phase space and the latter is subject to the Vlasov equation~\cite{H}.

For elastic matter a complete system on the matter fields can be obtained by assigning the energy density and the radial/tangential pressure as functions of the principal stretches of the matter, which are the eigenvalues of the deformation gradient, see~\cite{KS} and the Appendix below. Namely, letting $x(r)$ be the radial stretch and $y(r)$ the tangential stretch, and given three functions
\[
\hat{\rho},\ \hat{p}_\mathrm{rad},\ \hat{p}_\mathrm{tan}:[0,\infty)\times [0,\infty)\to [0,\infty),
\]
we set
\begin{equation}\label{eqstateh1h2}
\rho(r)=\hat{\rho}(x(r),y(r)),\quad p_\mathrm{rad}(r)= \hat{p}_\mathrm{rad}(x(r),y(r)),\quad  p_\mathrm{tan}(r)= \hat{p}_\mathrm{tan}(x(r),y(r)).
\end{equation}
The principal stretches $x$, $y$ are subject to the equations
\begin{subequations}\label{h1h2m}
\begin{equation}\label{eqh1}
x'=\frac{1}{\partial_x\hat{p}_\mathrm{rad}}\Big[-(\hat{p}_\mathrm{rad}+\hat{\rho})\frac{m+r^3\hat{p}_\mathrm{rad}/2}{r(r-2m)}+\frac{2}{r}(\hat{p}_\mathrm{tan}-\hat{p}_\mathrm{rad})-y'\partial_y\hat{p}_\mathrm{rad} \Big](x,y)\end{equation}
\begin{equation}\label{eqh2}
y'=\frac{2}{r}\left(\sqrt{\frac{r}{r-2m}\,xy}-y\right),
\end{equation}
while the Hawking mass satisfies
\begin{equation}\label{meq}
m'=\frac{1}{2}r^2\hat{\rho}(x,y).
\end{equation}
\end{subequations}
The form of the functions $\hat{\rho}, \hat{p}_\mathrm{rad}, \hat{p}_\mathrm{tan}$ is determined by the {\it constitutive equation} of the elastic body. In this paper we consider the class of John elastic materials introduced in~\cite{T}, for which we have
\begin{subequations}\label{consteq}
\begin{equation}\label{rho}
\hat{\rho}(x,y)=\hat{\rho}_0(n)\left[1+k\left(\frac{x+2y}{n^{2/3}}-3\right)\right],
\end{equation}
\begin{equation}\label{prad}
\hat{p}_\mathrm{rad}(x,y)=\hat{p}_0(n)\left[1+k\left(\frac{x+2y}{n^{2/3}}-3\right)\right]+\hat{\rho}_0(n)\frac{4k(x-y)}{3n^{2/3}},
\end{equation}
\begin{equation}\label{ptan}
\hat{p}_\mathrm{tan}(x,y)=\hat{p}_0(n)\left[1+k\left(\frac{x+2y}{n^{2/3}}-3\right)\right]-\hat{\rho}_0(n)\frac{2k(x-y)}{3n^{2/3}},
\end{equation}
where
\begin{equation}\label{n}
n=y\sqrt{x},\quad \hat{p}_0(n)=n^2\frac{d}{dn}\left(\frac{\hat{\rho}_0(n)}{n}\right),
\end{equation}
\end{subequations}
see the Appendix. The constitutive equations~\eqref{consteq} contain two free functions, $\hat{\rho}_0$ and $\hat{p}_0$, and a constant $k\geq 0$, which is called {\it rigidity parameter} (or shear parameter). By letting $k\to 0$ the elastic body becomes a perfect fluid with energy density and pressure given by
\[
\rho_0(r)=\hat{\rho}_0(n(r)),\quad p_0(r)=\hat{p}_0(n(r)),
\]
which we call the {\it underlying fluid}. Our results apply in particular when the underlying fluid has a linear or polytropic equation of state. Jonh materials belong to the class of perfect {\it quasi Hookean } materials introduced by Carter and Quintana to study the deformations of neutron stars crusts~\cite{CQ}. Unfortunately a throughout discussion of which elastic materials may be relevant for the applications in neutron stars physics seems to be missing in the literature. We choose John elastic materials as a case study to exemplify the rich diversity of properties possessed by elastic bodies in general relativity compared to other matter models, such as perfect fluids.   

The paper continues as follows. In Section~\ref{properties} we study the question of existence of a solid surface for elastic balls. In particular, a characterization of elastic balls with solid surface  is given when the underlying fluid has a linear equation of state. Another important issue discussed in Section~\ref{properties} is the validity of the dominant energy condition. In Section~\ref{numerics} we present numerical simulations indicating the existence of elastic matter distributions with several different shapes (balls, single and multiple shells) and compare our results with those obtained in~\cite{AR} for Vlasov matter distributions.

\section{Properties of regular ball solutions}\label{properties}
The first result proved in this section (Proposition~\ref{y>x})  applies to all constitutive equations for the elastic material verifying a number of inequalities. Among these we assume that $\hat{\rho}, \hat{p}_\mathrm{rad}, \hat{p}_\mathrm{tan}$ satisfy the following conditions:
\begin{itemize}
\item[(C1)] $\hat{\rho}, \hat{p}_\mathrm{rad}, \hat{p}_\mathrm{tan}\in C^1([0,\infty)\times[0,\infty))$;
\item[(C2)] $\hat{\rho}(x,y)>0$, for all $x,y>0$ and $\hat{\rho}(0,y)=\hat{\rho}(x,0)=0$, for all $x,y\geq 0$;
\item[(C3)] $\hat{p}_\mathrm{rad}(0,y)=\hat{p}_\mathrm{rad}(x,0)=\hat{p}_\mathrm{tan}(0,y)=\hat{p}_\mathrm{tan}(x,0)=0$, for all $x,y\geq 0$.
\end{itemize}
Additional conditions will be required when they are needed.
%
\begin{definition}
A triple $(x,y,m)$ of radial functions is called a regular ball solution of the system~\eqref{h1h2m} with radius $R>0$ and ADM mass $M>0$ if the following holds:
\begin{itemize}
\item[(i)] $x,y,m\in C^1([0,R])$;
\item[(ii)] $x(0)=y(0):=h>0$, $m(0)=0$;
\item[(iii)] $x,y>0$, for $x\in [0,R)$, and $0<2m(r)<r$, for all $r> 0$;
\item[(iv)] For $r>R$ we have $x=y=0$  and $m=m(R):=M$;
\item[(v)] $p_\mathrm{rad}(r):=\hat{p}_\mathrm{rad}(x(r),y(r))$ satisfies $p_\mathrm{rad}(r)>0$, $r\in [0,R)$, $p_\mathrm{rad}(R)=0$;
\item[(vi)] $(x,y,m)$ solves the system~\eqref{h1h2m}, for $r\in (0,R)$.
\end{itemize}
\end{definition}
We recall that $f\in C^1([a,b])$ means that $f\in C^1((a,b))$ and that $f,f'$ extend continuously on the boundary.
We also remark that, by (iv) and (C2)-(C3), the matter fields vanish for $r>R$ and
\[
\rho(r):=\hat{\rho}(x(r),y(r))>0,\quad \text{for $r\in[0,R)$}.
\]
In the exterior of the body  the spacetime metric is given by the Schwarzschild metric, cf.~\eqref{metricST}. Since $p_\mathrm{rad}$ and $m$ are continuous at the boundary, the usual matching conditions (continuity of the first and second fundamental form) are satisfied.
Whether $\rho$ vanishes or not at the boundary depends on the behavior of the principal stretches at $r=R$.
At the boundary of a regular ball solution, the principal stretches satisfy $(x(R),y(R))=(x_*,y_*)$, where
\begin{equation}\label{boundaryeq}
\hat{p}_\mathrm{rad}(x_*,y_*)=0.
\end{equation}
By (C3), the latter equation is always satisfied when either of the principal stretches vanishes. However only the radial stretch can vanish at the boundary. In fact, by~\eqref{eqh2} we have $y'\geq -2y/r$, which gives
\begin{equation}\label{ymin}
\inf_{r\in[0,R]}y(r)>0.
\end{equation}
Therefore, by (C2),
\[
\rho(R)=0 \Leftrightarrow x(R)=0.
\]
\begin{proposition}\label{y>x}
Consider a regular ball solution with radius $R>0$. Assume that, in addition to (C1)--(C3), the constitutive equations satisfy the following assumptions for all $h>0$:
\begin{itemize}
\item[(C4)] $\partial_x \hat{p}_\mathrm{rad}(h,h)>0$;
\item[(C5)] $\partial_x\hat{p}_\mathrm{rad}(h,h)+\partial_y\hat{p}_\mathrm{rad}(h,h)\geq 0$;
\item[(C6)] There exists $C>0$ such that $\hat{p}_\mathrm{tan}-\hat{p}_\mathrm{rad}\sim C(y-x)$, as $(x,y)\to (h,h)$.
\end{itemize}
Then the inequality $y>x$ holds for all $r\in (0,R)$.
\end{proposition}
\noindent\textbf{Remark 1: }The meaning of $A(z)\sim B(z)$ as $z\to a$ is that $A(z)/B(z)\to 1$ as $z\to a$. 

\noindent\textbf{Remark 2: }The hypotheses (C4)--(C6) are satisfied by (a subclass of) John materials, cf. Corollary~\ref{cor}.
\begin{proof}
Since $x(0)=y(0)=h>0$ we obtain
\[
y'\sim \frac{2h}{r}\frac{m(r)}{r}= \frac{h}{r^2}\int_0^r s^2\rho(s)\,ds,\quad \text{as $r\to 0$}.
\]
Hence
\begin{equation}\label{yprimeas}
\frac{y'(r)}{r}\sim \frac{1}{3}\rho(0) h,\quad \text{as $r\to 0$}.
\end{equation}
Furthermore,
\begin{align}
x'\sim& -a\,\frac{m(r)+r^3p_\mathrm{rad}(0)/2}{r(r-2m(r))}\nonumber\\
&+b\,\frac{y-x}{r}-y'\frac{\partial_y\hat{p}_\mathrm{rad}(h,h)}{\partial_x\hat{p}_\mathrm{rad}(h,h)},\quad\text{as $r\to 0$.}\label{xpas}
\end{align}
where $a,b>0$ are given by
\[
a=\frac{p_\mathrm{rad}(0)+\rho(0)}{\partial_x\hat{p}_\mathrm{rad} (h,h)}, \quad b=\frac{2C}{\partial_x\hat{p}_\mathrm{rad}(h,h)}
\]
In addition,
\[
\frac{m(r)+r^3p_\mathrm{rad}(0)/2}{r(r-2m(r))}\sim \frac{1}{2}(p_\mathrm{rad}(0)+\rho(0)/3)\, r.
\]
Hence, letting $\delta(r)=y(r)-x(r)$ we obtain
\begin{equation}\label{deltapas}
\delta'(r)\sim c_1 r- c_2\frac{\delta(r)}{r},\quad \text{as $r\to 0$,}
\end{equation}
where $c_1,c_2$ are positive constants.
Since $\delta(0)=0$,~\eqref{deltapas} implies
\[
\delta(r)\to 0^+,\quad \text{as $r\to 0$.}
\]
In particular $\delta$ is positive, i.e., $y>x$, for small $r$. Now suppose that there exists $r_*\in (0,R)$ such that $\delta(r_*)=0$. Since $\hat{p}_\mathrm{rad}=\hat{p}_\mathrm{tan}$ when $\delta=0$, we get
\begin{align*}
\delta'(r_*)=\ &\frac{1}{\partial_x\hat{p}_\mathrm{rad}(x(r_*),x(r_*))}\left[(p_\mathrm{rad}(r_*)+\rho(r_*))\frac{m(r_*)+r_*^3 p_\mathrm{rad}(r_*)/2}{r_*(r_*-2m(r_*))}\right]\\
&+\frac{2y(r_*)}{r_*}\left(\sqrt{\frac{r_*}{r_*-2m(r_*)}}-1\right)\left[1+\frac{\partial_y\hat{p}_\mathrm{rad}(x(r_*),y(r_*))}{\partial_x\hat{p}_\mathrm{rad}(x(r_*),y(r_*))}\right].
\end{align*}
The latter identity and (C4)-(C5) imply $\delta'(r_*)>0$, which is impossible. Thus $\delta(r)>0$ for all $r\in (0,R)$.
\end{proof}
In the rest of the paper we work with the constitutive equations~\eqref{consteq} of John elastic materials and we make the following assumptions on the equation of state of the underlying fluid:
\begin{itemize}
\item[(F1)] $\hat{\rho}_0\in C^1([0,\infty))$, $\hat{\rho}_0(n)\geq 0$ and equality holds iff $n=0$. Moreover $\hat{\rho}_0(n)/ n^{2/3}$ is bounded as $n\to 0$;
\item[(F2)]  $\hat{p}_0\in C^1([0,\infty))$, $\hat{p}_0(n)\geq 0$, $\hat{p}_0'(n)\geq 0$ and the equalities hold iff $n=0$. Moreover $\hat{p}_0(n)/ n^{2/3}$ is bounded as $n\to 0$;
\item[(F3)] The dominant energy condition holds: $\hat{p}_0(n)\leq \hat{\rho}_0(n)$, for all $n\geq 0$.
\end{itemize}
Note that when $k=0$ a pair $(x_*,y_*)$ is a solution of~\eqref{boundaryeq} if and only if $x_*$ and/or $y_*$ vanish, hence the underlying fluid does {\it not } have a solid boundary.
Examples of underlying fluids that satisfy the above hypothesis are some barotropic fluids with linear or polytropic equation of state. In the former case we have
\[
p_0(r)=(\Gamma-1)\rho_0(r),\quad 1<\Gamma\leq 2,
\]
which, by~\eqref{n}, corresponds to choose
\begin{equation}\label{linear}
\hat{\rho}_0(n)=a n^\Gamma,\quad \hat{p}_0(n)=a (\Gamma-1)n^\Gamma,
\end{equation}
where $a$ is an arbitrary positive constant. For a polytropic equation of state we have
\[
\rho_0(r)=\frac{1}{\Gamma-1}\left(\omega^{\frac{1}{\Gamma-1}}p_0(r)^\frac{1}{\Gamma}+p_0(r)\right),\quad 1<\Gamma\leq 2,\ \omega>0,
\]
which corresponds to a pair $(\hat{\rho}_0,\hat{p}_0)$ given by
\begin{equation}\label{polytropic}
\hat{\rho}_0(n)=\frac{1}{\Gamma-1}\left(\omega^{\frac{1}{\Gamma-1}}a^\frac{1}{\Gamma}n+an^\Gamma\right),\quad \hat{p}_0(n)=an^\Gamma,
\end{equation}
where again  $a$ is an arbitrary positive constant.

It is clear that John materials satisfy (C1), (C3). Moreover, letting
\[
\sigma=\frac{x+2y}{n^{2/3}}-3,
\]
so that $\hat{\rho}(x,y)=\hat{\rho}_0(n)(1+k\sigma)$,
and using that $\sigma\geq 0$ (with equality if and only if $x=y$), we conclude that John materials satisfy (C2) as well. Next we show that (C4)--(C6) in Proposition~\ref{y>x} are also satisfied. For this we need the following calculus result.

\begin{lemma}\label{calculus}
John elastic materials satisfy
\[
\hat{p}_\mathrm{tan}(x,y)-\hat{p}_\mathrm{rad}(x,y)=2k\frac{\hat{\rho}_0(n)}{n^{2/3}}(y-x),
\]
\[
\hat{p}_\mathrm{rad}(x,y)+2\hat{p}_\mathrm{tan}(x,y)=3\frac{\hat{p}_0(n)}{\hat{\rho}_0(n)}\,\hat{\rho}(x,y),
\]
\begin{align*}
\partial_x\hat{p}_\mathrm{rad}(x,y)=&\frac{4 k}{3}\frac{y}{x}\frac{\hat{\rho}_0(n)}{n^{2/3}}+\frac{1}{2x}\left(\frac{n\hat{p}_0'(n)}{\hat{p}_0(n)}\right)\hat{p}_\mathrm{rad}(x,y)\\
&+\frac{2k}{3}\frac{x-y}{x}\left[\frac{2\hat{p}_0(n)}{n^{2/3}}+\frac{\hat{\rho}_0(n)}{n^{2/3}}\left(\frac{7}{3}-\frac{n\hat{p}_0'(n)}{\hat{p}_0(n)}\right)\right],
\end{align*}
\begin{align*}
\partial_y\hat{p}_\mathrm{rad}(x,y)=&-\frac{4 k}{3}\frac{\hat{\rho}_0(n)}{n^{2/3}}+\frac{1}{y}\left(\frac{n\hat{p}_0'(n)}{\hat{p}_0(n)}\right)\hat{p}_\mathrm{rad}(x,y)\\
&+\frac{2k}{3}\frac{x-y}{y}\left[\frac{\hat{p}_0(n)}{n^{2/3}}+2\frac{\hat{\rho}_0(n)}{n^{2/3}}\left(\frac{1}{3}-\frac{n\hat{p}_0'(n)}{\hat{p}_0(n)}\right)\right].
\end{align*}
\end{lemma}

Assuming (F1), the validity of (C6) for John elastic materials follows immediately from the first identity in Lemma~\ref{calculus}.
Moreover
\[
\partial_x\hat{p}_\mathrm{rad}(x,x)=\frac{1}{x}\left[\frac{4k}{3}\hat{\rho}_0(x^{3/2})+\frac{1}{2\sqrt{x}}\,\hat{p}_0'(x^{3/2})\right],
\]
\[
\partial_y\hat{p}_\mathrm{rad}(x,x)=\frac{1}{x}\left[-\frac{4k}{3}\hat{\rho}_0(x^{3/2})+\frac{1}{\sqrt{x}}\,\hat{p}_0'(x^{3/2})\right],
\]
by which (C4) and (C5) follow as well in virtue of (F1)--(F3).
Hence Proposition~\ref{y>x} gives the following result.
\begin{corollary}\label{cor}
For John elastic materials with underlying fluid satisfying the assumptions (F1)--(F3), regular ball solutions satisfy the inequality $y(r)>x(r)$, for all $r\in (0,R)$.
\end{corollary}
We also remark that by the second identity of Lemma~\ref{calculus}, and (F1)--(F3), we obtain the estimate
\begin{equation}\label{Omega}
p_\mathrm{rad}(r)+2p_\mathrm{tan}(r)\leq \Omega\rho(r),\quad \Omega=\sup_{n\geq 0}\frac{3\hat{p}_0(n)}{\hat{\rho}_0(n)}\leq 3,
\end{equation}
and thus the result proved in~\cite{H2} gives the following bound
\begin{equation}\label{buchdahl}
\frac{2m(r)}{r}\leq \frac{(1+2\Omega)^2-1}{(1+2\Omega)^2}\leq\frac{48}{49}.
\end{equation}
Next we discuss the validity of the dominant energy condition.
\begin{proposition}\label{matterfields}
Consider a John elastic material with underlying fluid satisfying (F1)--(F3). Let $(x,y,m)$ be a regular solution of~\eqref{h1h2m} and let $\rho(r),p_\mathrm{rad}(r),p_\mathrm{tan}(r)$ be the energy density, radial pressure and tangential pressure of the elastic body defined through~\eqref{eqstateh1h2}.  Then, for all $k>0$,
\begin{itemize}
\item[(i)] The dominant energy condition is satisfied in the radial direction and
\[
0<\frac{p_\mathrm{rad}(r)}{\rho(r)}<\frac{p_0(r)}{\rho_0(r)},\quad \text{for all $r\in (0,R)$}.
\]
\item[(ii)] Let $k<1$ and $\beta=\max(k,1/3)$. If
\[
\hat{p}(n)\leq (1-\beta)\hat{\rho}(n),
\]
the dominant energy condition is satisfied in the tangential direction and
\[
0<\frac{p_\mathrm{tan}(r)}{\rho(r)}<\frac{p_0(r)}{\rho_0(r)}+\beta,\quad \text{for all $r\in (0,R)$}.
\]
\item[(iii)] If
\[
\lim_{r\to R^+}\frac{3p_0(r)}{2\rho_0(r)}>1,
\]
then the dominant energy condition in the tangential direction is violated in an interior neighborhood of the boundary, i.e., there exists $\varepsilon>0$ such that $p_\mathrm{tan}(r)>\rho(r)$, for $r\in (R-\varepsilon,R]$.
\end{itemize}
\end{proposition}
\begin{proof}
Using $y>x>0$ in the definition of $\hat{p}_\mathrm{rad}$ we obtain
\[
\hat{p}_\mathrm{rad}(x,y)< \hat{p}_0(n)\left[1+k\left(\frac{x+2y}{n^{2/3}}-3\right)\right]=\frac{\hat{p}_0(n)}{\hat{\rho}_0(n)} \hat{\rho}(x,y).
\]
To prove (ii) we use that
\begin{equation}\label{tempo}
\frac{2k}{3}\frac{y-x}{n^{2/3}}< \beta(1+k\sigma),\quad\text{for all $y>x>0$.}
\end{equation}
To establish the latter inequality, consider the polynomial
\[
\mathcal{P}(z)=2k(\beta-1/3)z^3+\beta(1-3k)z^2+k(\beta+2/3).
\]
Then, letting $z=(y/x)^{1/3}$,  the inequality~\eqref{tempo} is equivalent to $\mathcal{P}(z)> 0$, for $z> 1$, which holds in particular for $k<1$ and $\beta=\max(k,1/3)$. Using~\eqref{tempo} in the definition of $\hat{p}_\mathrm{tan}$ we find
\begin{align*}
\hat{p}_\mathrm{tan}&< \hat{p}_0(n)\left[1+k\left(\frac{x+2y}{n^{2/3}}-3\right)\right]+\beta\hat{\rho}_0(n)\left[1+k\left(\frac{x+2y}{n^{2/3}}-3\right)\right]\\
&=\left(\frac{\hat{p}_0(n)}{\hat{\rho}_0(n)}+\beta\right)\hat{\rho}(x,y).
\end{align*}
The claim (iii) follows by the identity
\[
\left(\frac{\hat{p}_\mathrm{tan}}{\hat{\rho}}\right)_{\hat{p}_\mathrm{rad}=0}=\frac{3\hat{p}_0(n)}{2\hat{\rho}_0(n)}.
\]
\end{proof}
Applying Proposition~\ref{matterfields} to the case of an underlying fluid with linear equation of state we obtain the following result.
\begin{corollary}\label{eclin}
Assume that the underlying fluid of the John elastic material has the linear equation of state $p_0=(\Gamma-1)\rho_0$, $1<\Gamma\leq 2$. Then
\begin{itemize}
\item[(i)] If $k<1$ and $1<\Gamma\leq 2-\max(k,1/3)$, the elastic body satisfies the dominant energy condition;
\item[(ii)]if $\Gamma>5/3$, the dominant energy condition in the tangential direction is violated in an interior neighborhood of the boundary.
\end{itemize}
\end{corollary}
The range where Corollary~\ref{eclin} claims the validity/violation of the dominant energy condition is depicted in Figure~\ref{energycondfig}. We leave open the question when $\max(1,2-k)<\Gamma\leq 5/3$. 
\begin{figure}[Ht!]
\begin{center}
\includegraphics[trim= 40mm 16cm 0mm 45mm, clip=true, width=1.5\textwidth]{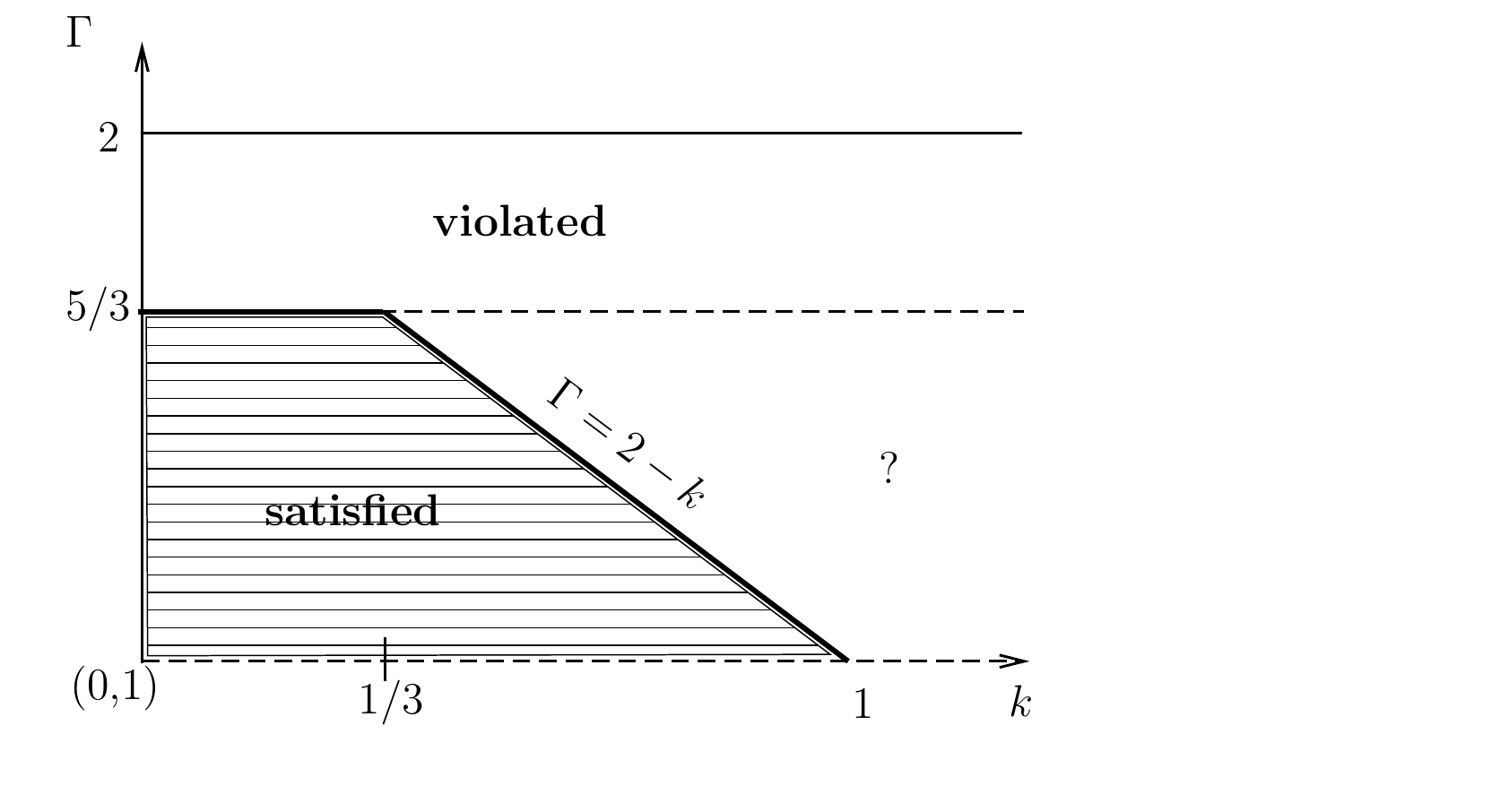}
\end{center}
\caption{The values of the parameters that entail validity/violation of the dominant energy condition. $k$ is the rigidity parameter of the elastic body and $\Gamma$ is the parameter in the linear equation of state $p_0=(\Gamma-1)\rho_0$ of the underlying fluid.}
\label{energycondfig}
\end{figure}

To conclude this section we study the behavior at the boundary of the energy density when the underlying fluid has a linear equation of state.

\begin{theorem}\label{leqs}
Assume that the underlying fluid of the John material has the linear equation of state $p_0=(\Gamma-1)\rho_0$, $1<\Gamma\leq 2$. Then a regular ball solution of~\eqref{h1h2m} satisfies $x(R)>0$, or equivalently $\rho(R)>0$, if and only if $\Gamma,k$ satisfy one of the following two conditions:
\begin{itemize}
\item[(i)] $1<\Gamma<5/3$;
\item[(ii)] $\Gamma\geq 5/3$ and $k>1/3$.
\end{itemize}
Moreover the mass and the radius of regular ball solutions with $\rho(R)>0$ satisfy the inequality
\begin{equation}\label{MRlin}
\frac{2M}{R}\geq \frac{6(\Gamma-1)}{6\Gamma-5}
\end{equation}
\end{theorem}
\begin{proof}
When $p_0=(\Gamma-1)\rho_0$, the solutions of~\eqref{boundaryeq} lie on the straight lines $x_*=0$ and $y_*=\alpha^3 x_*$ where $\alpha$ is the positive root of
\begin{equation}\label{alphaeq}
\mathcal{P}(\alpha)=2k\left(\Gamma-\frac{5}{3}\right)\alpha^3+(1-3k)(\Gamma-1)\alpha^2+k\left(\Gamma+\frac{1}{3}\right)=0.
\end{equation}
It is straightforward to show that the latter equation has a solution $\alpha >1$ if and only if $\Gamma$, $k$ satisfy one of the conditions (i)-(ii) in the theorem.
Moreover, since $\mathcal{P}(1)=\Gamma-1>0$, it follows that  the straight line $y_*=\alpha^3x_*$ lies on the left of the line $y=x$ and on the right of the axis $x=0$, see Figure~\ref{prad=0}. Since at the center we have $x(0)=y(0)$, and, by Proposition~\ref{y>x}, $y(r)>x(r)$, for all $r\in (0,R)$, it follows that for a regular ball solution the curve $r\to (x(r),y(r))$ cannot reach the line $x=0$ without first hitting the line $y=\alpha^3x$. This concludes the proof of the first part of the theorem. To prove the inequality~\eqref{MRlin} we use that, by the TOV equation~\eqref{TOVeq},
\[
\lim_{r\to R^-}p_\mathrm{rad}'=-\rho(R)\frac{M}{R(R-2M)}+2\frac{p_\mathrm{tan}(R)}{R}.
\]
Using the identity
\[
\frac{p_\mathrm{tan}(R)}{\rho(R)}=\left(\frac{\hat{p}_\mathrm{tan}}{\hat{\rho}}\right)_{\hat{p}_\mathrm{rad}=0}=\frac{3\hat{p}_0(n)}{2\hat{\rho}_0(n)}=\frac{3}{2}(\Gamma-1)
\]
we obtain, after straightforward calculations,
\[
\lim_{r\to R^-}p_\mathrm{rad}'=-\frac{\rho(R)}{2R}\left[\frac{Z}{1-Z}-6(\Gamma-1)\right],\quad Z=\frac{2M}{R}.
\]
Since $\lim_{r\to R^-}p_\mathrm{rad}'\leq 0$, we must have $Z/(1-Z)-6(\Gamma-1)\geq0$, which gives~\eqref{MRlin}.
%
%
\end{proof}
We emphasize that the previous theorem does {\it not} claim the existence of regular ball solutions. In the next section we present numerical evidence of their existence, provided $k$ is sufficiently large (depending on $\Gamma$). We also anticipate that $x(R)>0$ holds for all solutions found numerically. We conjecture that if elastic balls with $\rho(R)=0$ exist, they are non-generic, i.e., they correspond to an isolated set of initial data $h>0$.

Notice that~\eqref{MRlin} is not in contradiction with~\eqref{buchdahl}. In fact, when evaluated at $r=R$ and $\Omega=3(\Gamma-1)$, inequality~\eqref{buchdahl} gives
\[
\frac{2M}{R}\leq \frac{(6\Gamma-5)^2-1}{(6\Gamma-5)^2}
\]
and the right hand side of the latter inequality is greater than the right hand side of~\eqref{MRlin} when $\Gamma>1$. We also remark that the right hand side of~\eqref{MRlin} can be quite close to one. For instance, assuming $\Gamma=5/3$ (which is the maximum value that entails validity of the dominant energy condition) we get $2M/R\geq 0.8$, which corresponds to a redshift 
$z=\sqrt{5}-1\approx 1.24$, where we recall that 
\[
z=\frac{1}{\sqrt{1-\frac{2M}{R}}}-1.
\]
 This means that elastic bodies satisfying the dominant energy condition can be very small and dense objects (in fact, more dense than neutron stars, for which $z\approx 0.25-0.35$). One undesirable  feature of the inequality~\eqref{MRlin} is that it is independent of the rigidity parameter $k$. As shown in Figure~\ref{Z}, the exact value of $2M/R$ grows with $k$ and can get much closer to one than the right hand side of~\eqref{MRlin} and indeed redshifts $z\approx2.8$ can be attained without violation of the dominant energy condition.
\begin{figure}[Ht]
\begin{center}
\includegraphics[width=0.8\textwidth]{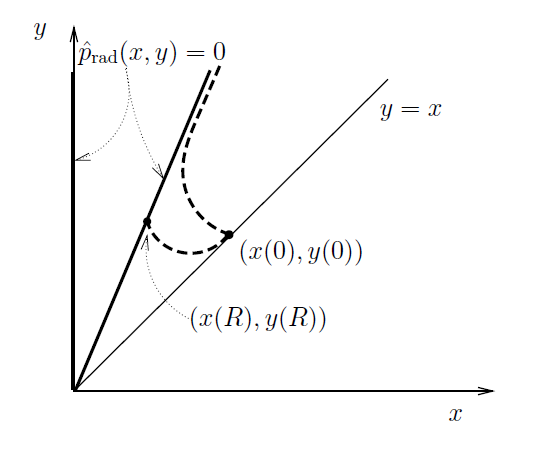}
\end{center}
\caption{This picture summarizes the argument in the proof of Theorem~\ref{leqs}. Under the assumptions on the parameters $k,\Gamma$, the curve $r\to(x(r),y(r))$ cannot intersect the axis $x=0$ without first hitting the portion of the graph $\hat{p}_\mathrm{rad}(x,y)=0$ where $x>0$. The other alternative is that the curve $r\to (x(r),y(r))$ does not intersect the graph $\hat{p}_\mathrm{rad}(x,y)=0$, in which case the body does not have finite radius.}
\label{prad=0}
\end{figure}

\section{Numerical results}\label{numerics}
The numerical simulations presented in this section have been performed using the ODE solver of MATHEMATICA 8. To obtain regular ball solutions  one has to integrate the system~\eqref{h1h2m} with initial datum at the center. Since the system is singular at $r=0$, we fix the initial datum at $r_*=10^{-9}$. The value of $(x,y,m)$ at $r=r_*$ has been obtained by Taylor expanding the solution up to second order. For instance $x(r_*)=h+x'(0)r_*+x''(0)r_*^2/2$, where we compute $x',x''$ using~\eqref{eqh1}. This procedure is not necessary in the case of shells, since these solutions are obtained by integrating numerically the system~\eqref{h1h2m} with initial datum away from the center. In the case of shells, the initial data for $x,y$ at $r=R_\mathrm{int}$ (the internal radius of the shell) are given by
\[
x(R_\mathrm{int})=x_*,\quad y(R_\mathrm{int})=y_*,
\]
where $\hat{p}_\mathrm{rad}(x_*,y_*)=0$, which has been solved numerically.  In all cases (balls, shells, etc.), the integration stops when the radial pressure vanishes.

For our numerical simulations we assume a John elastic material whose underlying fluid has either linear or polytropic equation of state. We consider the two cases in two different subsections.

\subsection{Underlying fluid with linear equation of state}
It is well known that for perfect fluids with linear equation of state, spherically symmetric static solutions with finite radius do not exist. Our numerical simulations show that by adding a sufficiently large strain to the fluid, we obtain a body with finite radius. Precisely, letting $p_0=(\Gamma-1)\rho_0$, $1<\Gamma\leq 2$, the equation of state of the underlying fluid, we have found that there exists $k_*=k_*(\Gamma)$  such that the elastic body has finite radius if $k>k_*$, see Figure~3. Moreover $k_*$ depends only on $\Gamma$ not on the initial datum $h=x(0)=y(0)$ at the center.

\vspace{1cm}

 \begin{minipage}{\textwidth}
  \begin{minipage}[b]{0.7\textwidth}
    \centering
    \includegraphics[width=1\textwidth]{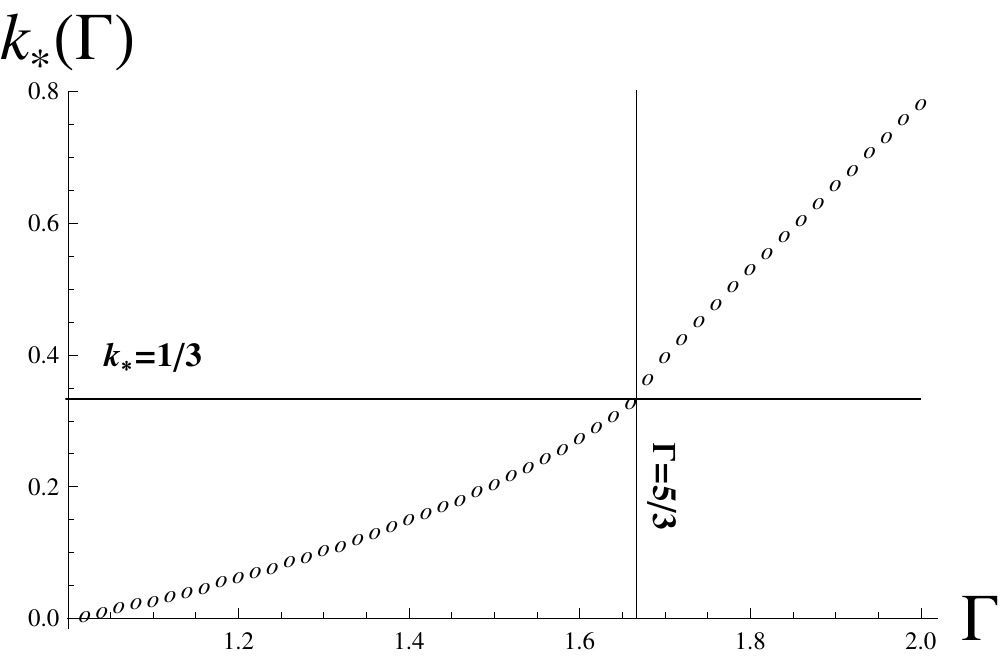}
    \captionof{figure}{Minimum value of the rigidity parameter $k$ to have an elastic body with finite radius when the underlying fluid has equation of state $p_0=(\Gamma-1)\rho_0$. Some values of $k_*(\Gamma)$ are given in the table beside. Note that $k_*(5/3)=1/3$.}
\label{plotsLEOS}
  \end{minipage}
  \hfill
  \begin{minipage}[b]{0.25\textwidth}
    \centering
    \begin{tabular}{|c|c|}\hline
& \\
      $\Gamma$ & $k_*(\Gamma)$ \\[0.1cm] \hline
&  \\
        6/5 & 0.0667 \\
        7/5 & 0.1533 \\
        3/2 & 0.2067 \\
        5/3 & 0.3333 \\
        9/5 & 0.6333 \\
        2 & 0.7833 \\ \hline
      \end{tabular}\\[1cm]
         \end{minipage}
  \end{minipage}
\vspace{1cm}

In Figure~\ref{profile} we show a typical profile of the matter variables for a regular ball solution. All the solutions that we found have a positive energy density at the boundary.

\begin{figure}[Ht!]
\begin{center}
\includegraphics[width=0.8\textwidth]{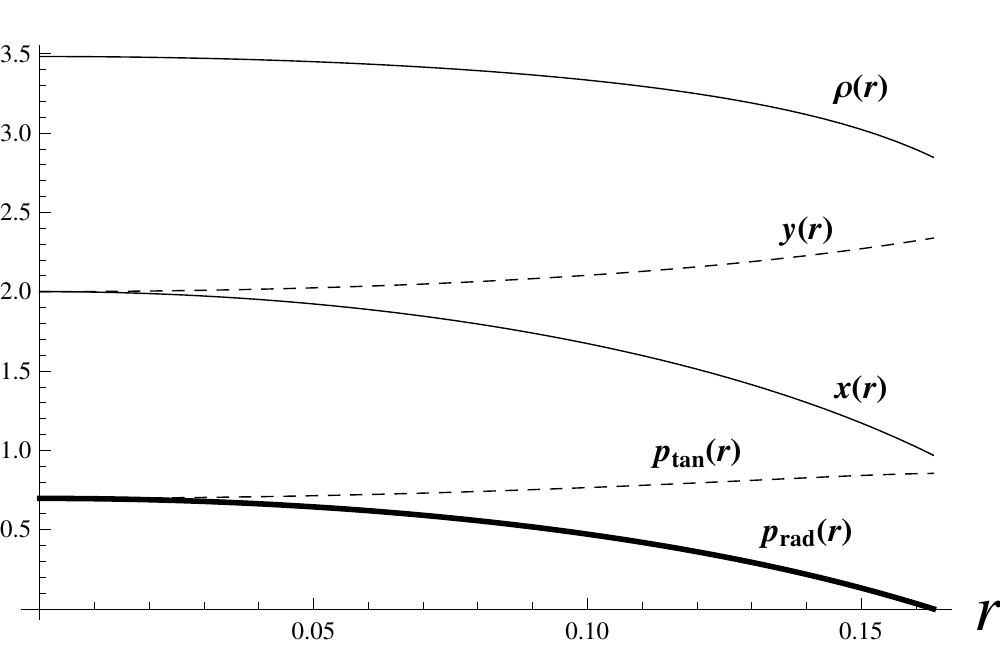}
\end{center}
\caption{The picture shows the profile of the matter fields and  of the principal stretches for an elastic body with underlying fluid with linear equation of state $p_0=(\Gamma-1)\rho_0$. We choose $\Gamma=6/5$ and $k=1/5$. Notice that $\rho(R)>0$. We have been unable to find numerical solutions with zero energy density on the boundary. }
\label{profile}
\end{figure}

In Figure~\ref{Z} we depict the value of $Z=2M/R$ for $\Gamma=5/3$ and several values of $k>k_*(5/3)=1/3$. We recall that $5/3$ is the largest value of $\Gamma$ such that, according to Corollary~\ref{eclin},  the dominant energy condition is satisfied at the boundary. In fact, we have found numerically that the dominant energy condition is satisfied everywhere in the support of the solution. Remarkably, $2M/R$ can attain rather large values, greater than the Buchdahl limit $8/9$. For large values of the rigidity parameter we have $Z\approx 0.93$, which corresponds to a redshift $z\approx2.78$. Thus elastic stars might have a very large gravitational redshift, without violating the dominant energy condition.

\begin{figure}[Ht]
\begin{center}
\includegraphics[width=0.8\textwidth]{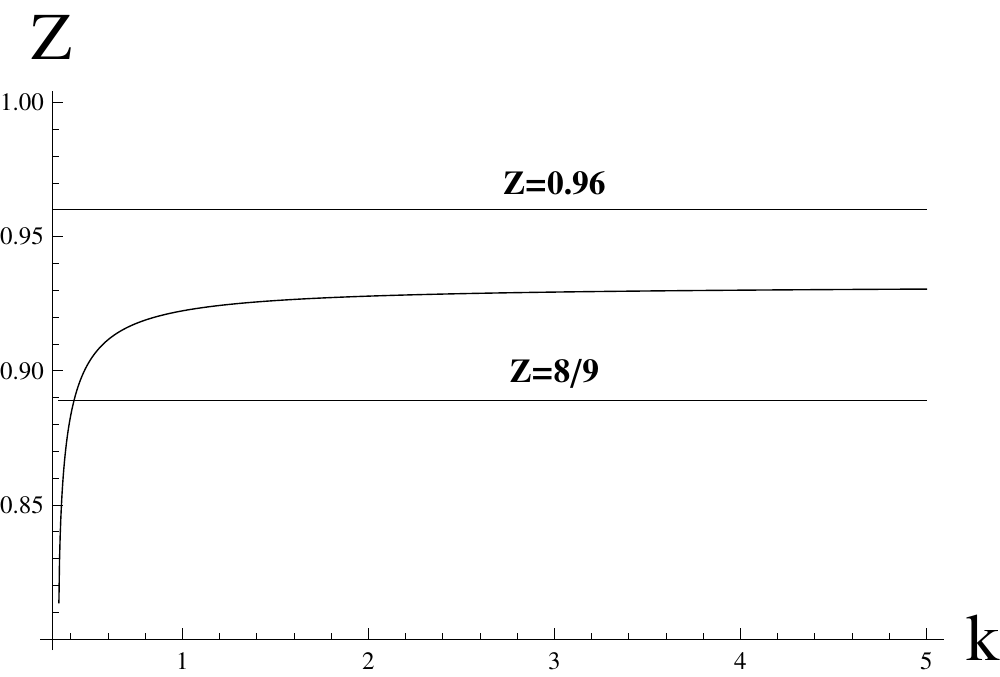}
\end{center}
\caption{This picture shows the value of $Z=2M/R$ in terms of the rigidity parameter for an underlying fluid with equation of state $p_0=2\rho_0/3$ (i.e., $\Gamma=5/3$). The profile of the curve $Z(k)$ is independent of the initial data. For $k>\bar{k}\sim 0.418$, the value of $Z$ dominates the Buchdahl limit $8/9$. The upper horizontal line corresponds to the dominating value $Z=0.96$ given by~\eqref{buchdahl}, which holds with $\Omega=2$  for the elastic material under consideration. }
\label{Z}
\end{figure}

\subsection{Underlying fluid with polytropic equation of state}
Let us first review the known results concerning the existence of regular fluid balls with poytropic equation of state~\eqref{polytropic}. The main question here is for which values of the polytropic index $q$ the solution of the TOV equation has finite radius, where
\[
q=\frac{1}{\Gamma-1}.
\]
This question has been studied extensively in the physical and mathematical literature, see for instance~\cite{HRU,Mak,RR} and references therein. The finite radius property is known to be verified for $0<q\leq 3$ for all data at the center and for $q\in (3,5)$ if $\rho(0)/(p(0)+\rho(0))$ is sufficiently small. For $q\geq 5$,  the polytropic fluid does not have finite radius.

Our numerical solutions indicate that within the class of elastic materials considered in this paper, polytropic perfect fluids without finite radius  are structurally unstable: when the polytropic index $q$ is equal to 5, the underlying fluid does not have finite radius while the elastic body does, no matter how small is the modulus of rigidity. Figure~\ref{kR} depicts (in a logarithmic scale) the radius of the elastic star as a function of the modulus of rigidity for $q=5$.

\begin{figure}[Ht]
\begin{center}
\includegraphics[width=0.8\textwidth]{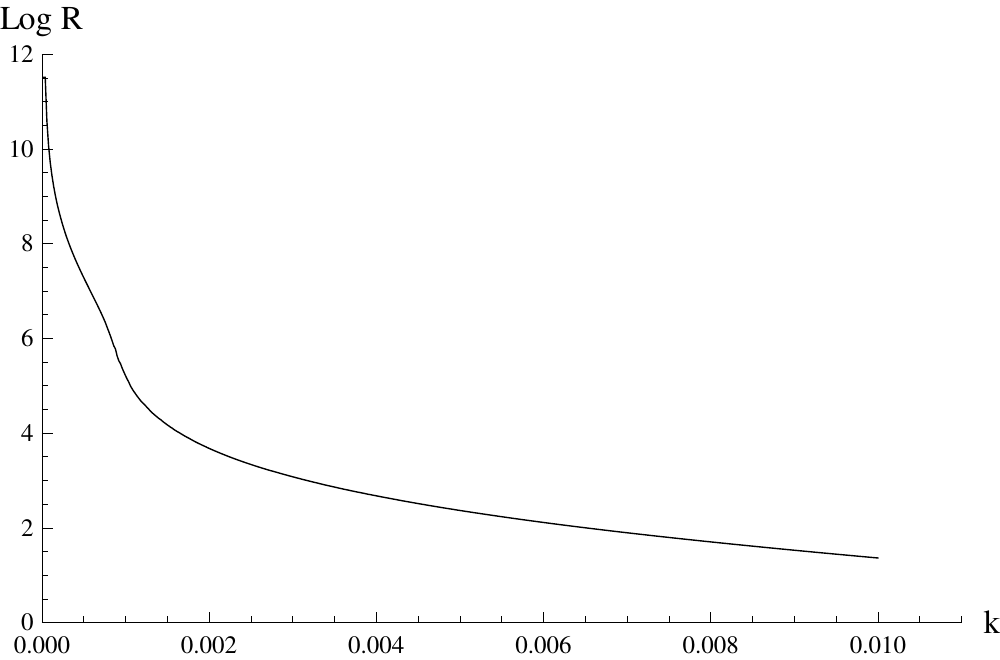}
\end{center}
\caption{The radius of the elastic body as a function of the rigidity parameter (in logarithmic scale). The poytropic index of the underlying fluid is $q=5$. }
\label{kR}
\end{figure}

In the rest of the simulations we assume that the underlying fluid has a polytropic equation of state with polytropic index $q=3/2$, and thus in particular it has finite radius. In Figures~\ref{sols1}-\ref{sols2} we depict solutions with several different shapes: balls, single shells with and without core, double shells. Note that $\rho=0$ at the boundary of the underlying fluid, while $\rho$ never vanishes at the boundary of the elastic solutions depicted in Figures~\ref{sols1}-\ref{sols2}. Moreover for the underlying fluid only ball-shaped static solutions are possible, because the pressure is a decreasing function of the radius. The many-body solutions depicted in Figures~\ref{sols1}-\ref{sols2} have been obtained by gluing together single bodies solutions through a vacuum region. There seems to be no restriction on the relative position of the shells, which is a property that distinguishes elastic matter from Vlasov matter, see the discussion in Section~\ref{vlasovsection}.

\begin{figure}[Ht!]
\subfigure{
\includegraphics[width=0.8\textwidth]{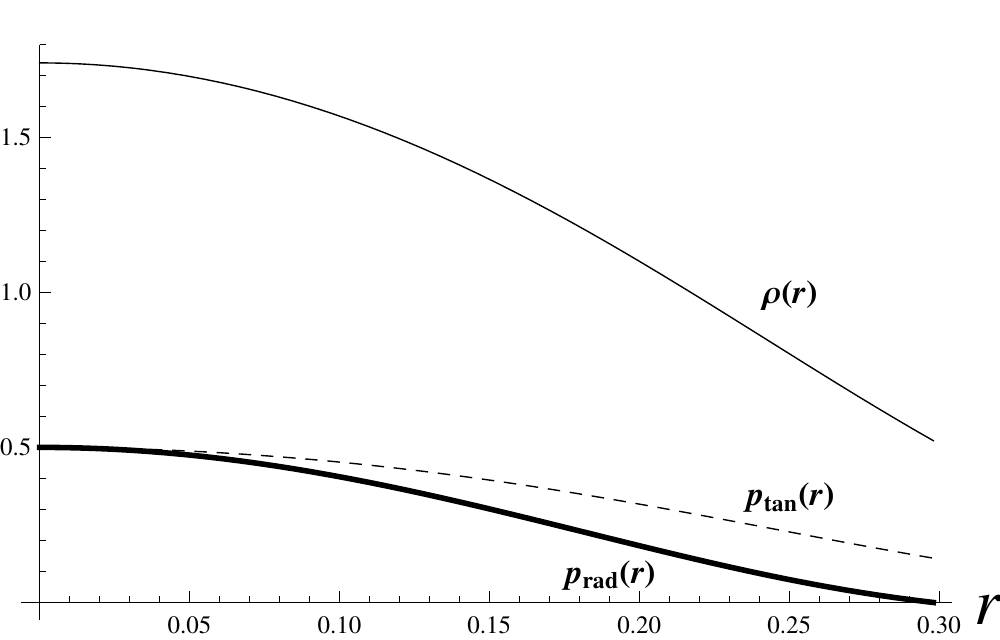}}\\
\subfigure{
\includegraphics[width=0.8\textwidth]{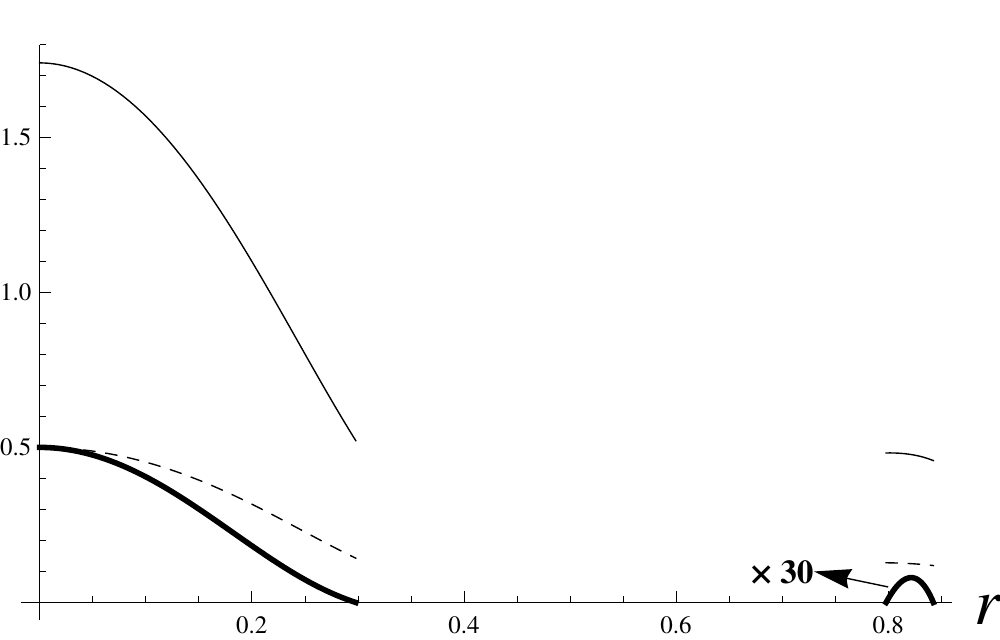}}
\caption{This picture shows an elastic ball and a ball surrounded by a shell ($k=1/10, q=3/2$). The radial pressure of the outer shell has been enhanced for graphical purposes.}
\label{sols1}
\end{figure}

\begin{figure}[Ht!]
\subfigure{
\includegraphics[width=0.8\textwidth]{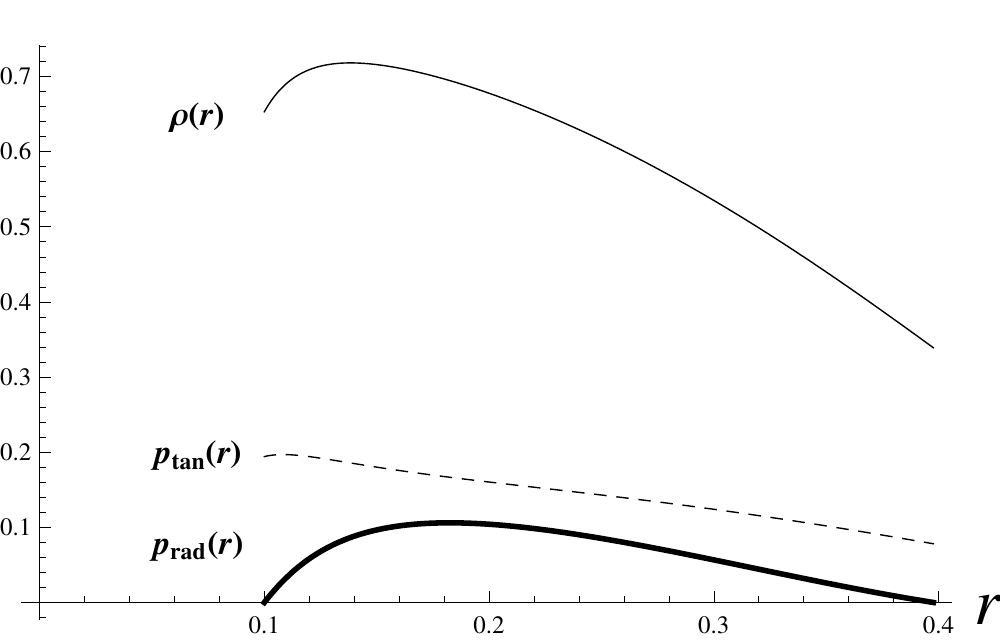}}\\
\subfigure{
\includegraphics[width=0.8\textwidth]{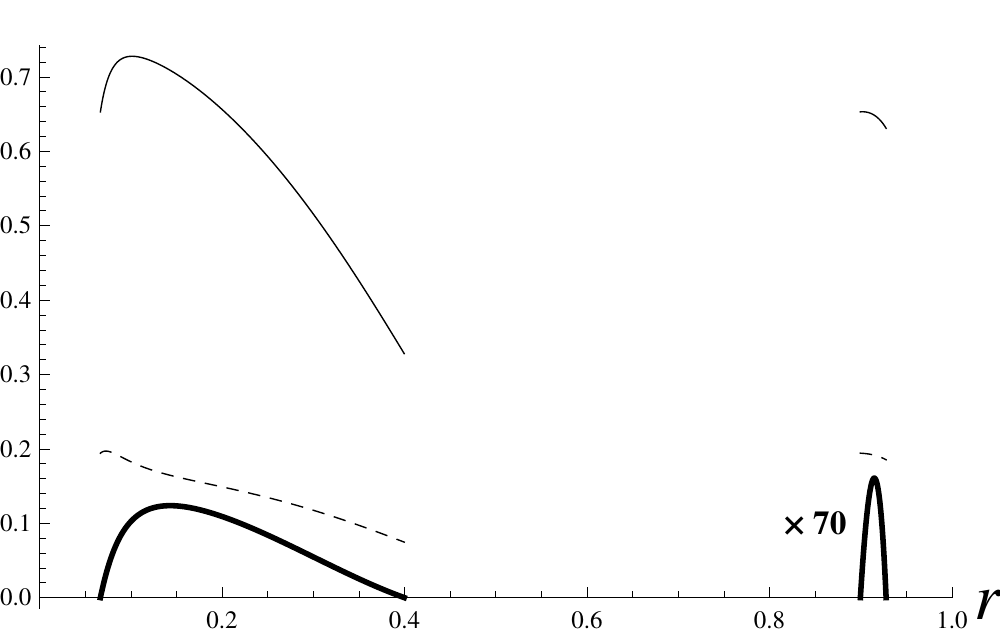}}
\caption{This picture shows a single and a double shell ($k=1/10, q=3/2$). The radial pressure of the outer shell has been enhanced for graphical purposes.}
\label{sols2}
\end{figure}

Finally in Figure~\ref{RM} we depict the mass vs radium diagram of ball solutions for different values of $k$. For $k=0$, which corresponds to a polytropic fluid, we recognize the characteristic spiral structure of the diagram. This structure is preserved for small value of the rigidity parameter. However our numerical simulations indicate that for $k\geq 1/3$ the spiral structure is no longer present.

\begin{figure}[Ht!]
\begin{center}
\includegraphics[width=1\textwidth]{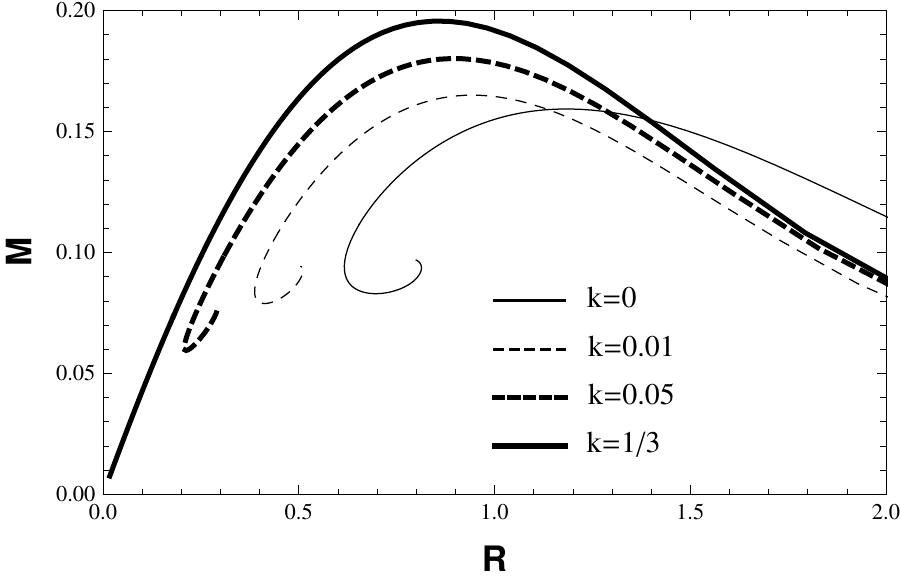}
\end{center}
\caption{The radius-mass diagram of ball solutions. The underlying fluid is polytropic with index $q=3/2$. The spiral structure, which appears for perfect fluids and Vlasov matter, is also present for elastic bodies only if the rigidity parameter is sufficiently small.}
\label{RM}
\end{figure}

\subsection{Comparison with Vlasov matter}\label{vlasovsection}
So far we underlined the differences between perfect fluids and elastic bodies, the former class of matter models arising in the limit $k\to 0$ in the constitutive equations of John materials. In this subsection we compare our results with those obtained in~\cite{AR} in the case of Vlasov matter. One common feature of elastic matter and Vlasov matter is that they are both anisotropic models, in contrast to the fluid model which is isotropic.

The fundamental quantity for describing Vlasov matter is the density function $f$ which is a non-negative function on phase space. The energy density $\rho$, the radial pressure $p_{\mathrm{rad}}$ and the tangential pressure $p_{\mathrm{tan}}$ are given in terms of $f$ by the following expressions in suitable coordinates
\begin{eqnarray}
\rho(r)&=&
\int_{\mathbb{R}^3}\sqrt{1+v^2}f(r,v)\;dv,\label{rho}\\
p_{\mathrm{rad}}(r)&=&\int_{\mathbb{R}^3}\frac{(x\cdot v)^{2}}{\sqrt{1+v^2}}f(r,v)\;dv,\label{p}\\
p_{\mathrm{tan}}(r)&=&\frac12\int_{\mathbb{R}^3}\frac{|x\times v|^2}{\sqrt{1+v^2}}f(r,v)\;
dv.\label{q}
\end{eqnarray}
Here $v\in \mathbb{R}^3$ can be thought of as the three-momentum of the particles, cf.~\cite{H} for more background. From~(\ref{p}) we conclude that $p_{\mathrm{rad}}(r)=0$ if and only if $f(r,\cdot)=0$. Hence, since $p_{\mathrm{rad}}(R)=0$ at the boundary $r=R$ of a static body, it follows that also $p_{\mathrm{tan}}$ and $\rho$ vanish at $r=R$. This is in sharp contrast to our results above for elastic matter where we found that $\rho(R)>0$ (cf. Figures~\ref{sols1} and~\ref{sols2}) for all solutions we construct numerically. This difference is perhaps intuitive since Vlasov matter models galaxies, which typically have a thin atmosphere close to the boundary, whereas elastic matter models very dense stars, which reasonably possess a solid surface.

Another distinct difference between Vlasov matter and elastic matter is the maximum possible value of the quantity $2M/R$. The expressions above imply that 
\begin{equation}\label{Vlasovcon}
p_{\mathrm{rad}}+2p_{\mathrm{tan}}\leq \rho.\nonumber
\end{equation}
Hence, Vlasov matter satisfies~(\ref{Omega}) with $\Omega=1$, which in view of~\eqref{buchdahl} implies that
\[
\frac{2m(r)}{r}\leq \frac89, \mbox{ for all }r\in [0,R].
\]
For elastic matter we found above examples where $2M/R\approx 0.93$, which corresponds to $z\approx 2.8$, hence quite larger than the Buchdahl limit $8/9$ corresponding to $z=2$.

An interesting feature of Vlasov matter is that static solutions may consist of many peaks as is shown in Figure~\ref{figvlasov1}. Note in particular that there is a vacuum region between the first and second peak (there are in fact cases where there are many such vacuum regions, cf.~\cite{AR}). These vacuum regions appear naturally in the case of Vlasov matter, i.e., the solution has not been pieced together by separate solutions, cf.~\cite{AR} for more background
 on this issue. This is in contrast to the solutions shown in Figures~\ref{sols1} and~\ref{sols2} and entails that, while for elastic matter it seems that the shells can be placed at any radius, for Vlasov matter the position of the peaks is determined by the solution.

The property that the mass-radius diagram for elastic matter contains spirals, as depicted in Figure~\ref{RM}, is also true in the case of Vlasov matter. However, the spirals are more pronounced in the latter case, i.e., the spirals consist of many turns as is shown in Figure~\ref{spiralsfig}.

\begin{figure}[Ht!]
\begin{center}
\includegraphics[width=1\textwidth]{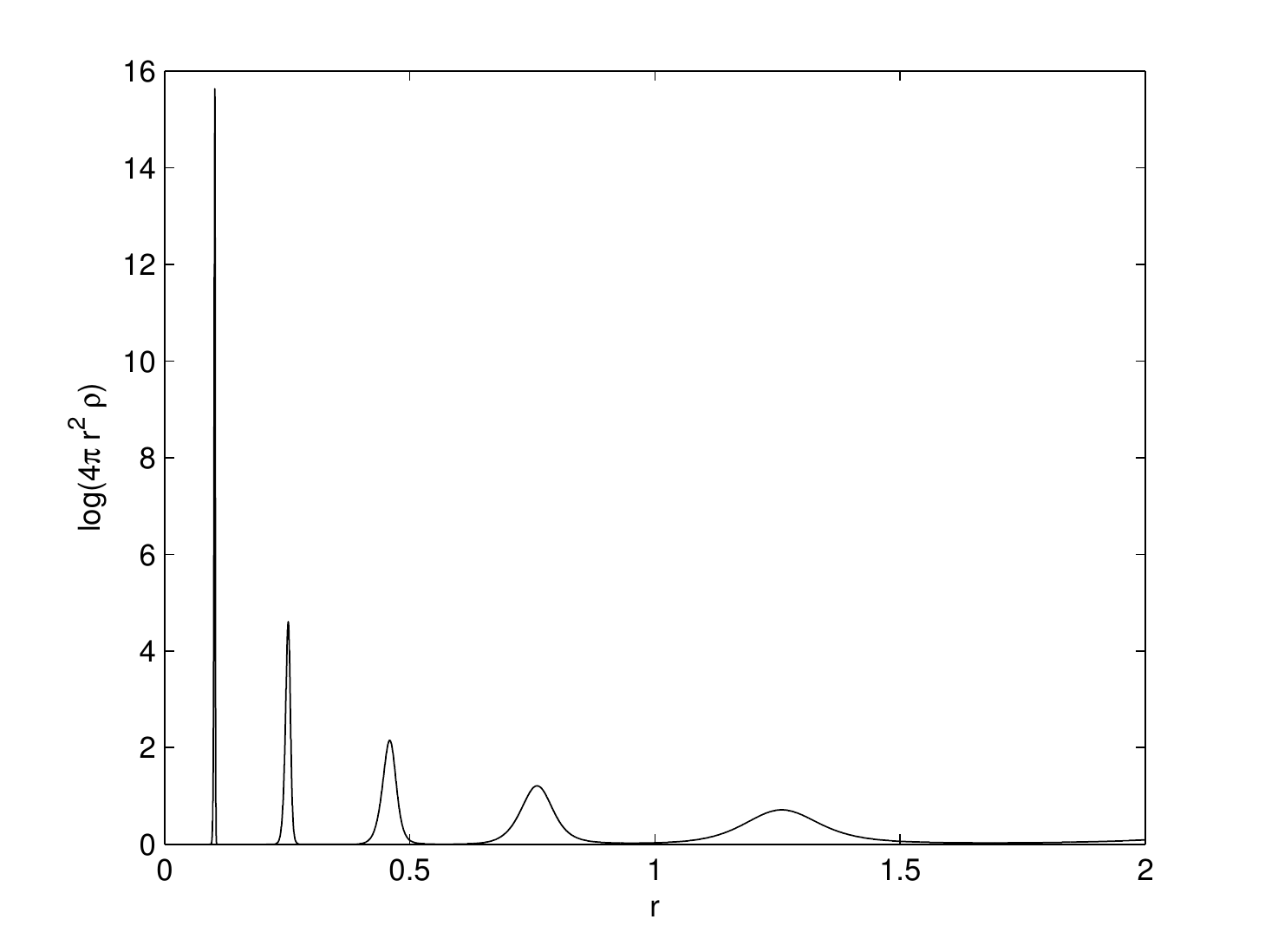}
\end{center}
\caption{Multi-peaks for Vlasov matter.}
\label{figvlasov1}
\end{figure}

\begin{figure}[Ht!]
\subfigure{\includegraphics[width=0.5\textwidth, height=5cm]{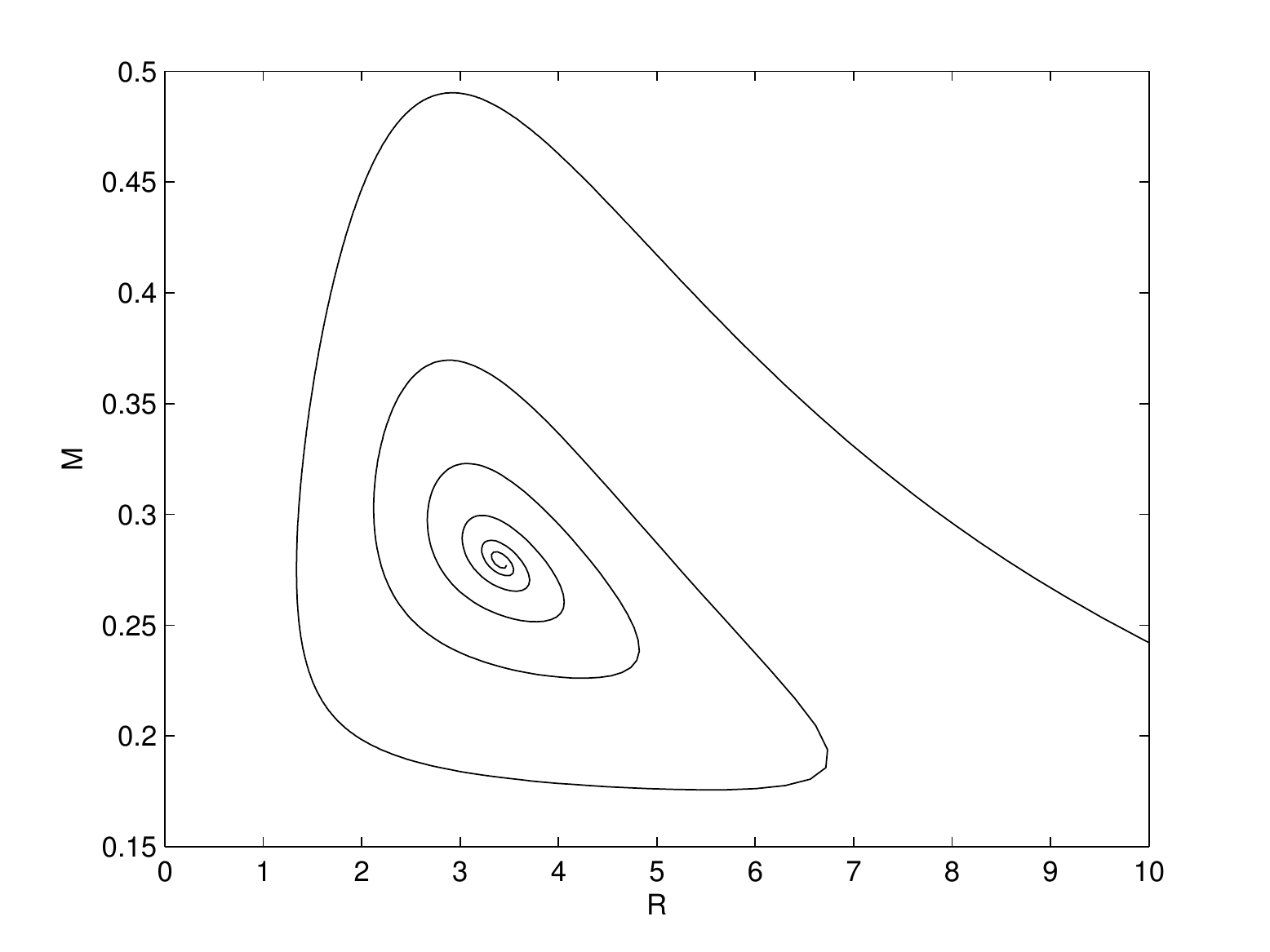}}\
\subfigure{
\includegraphics[width=0.47\textwidth, height=5cm]{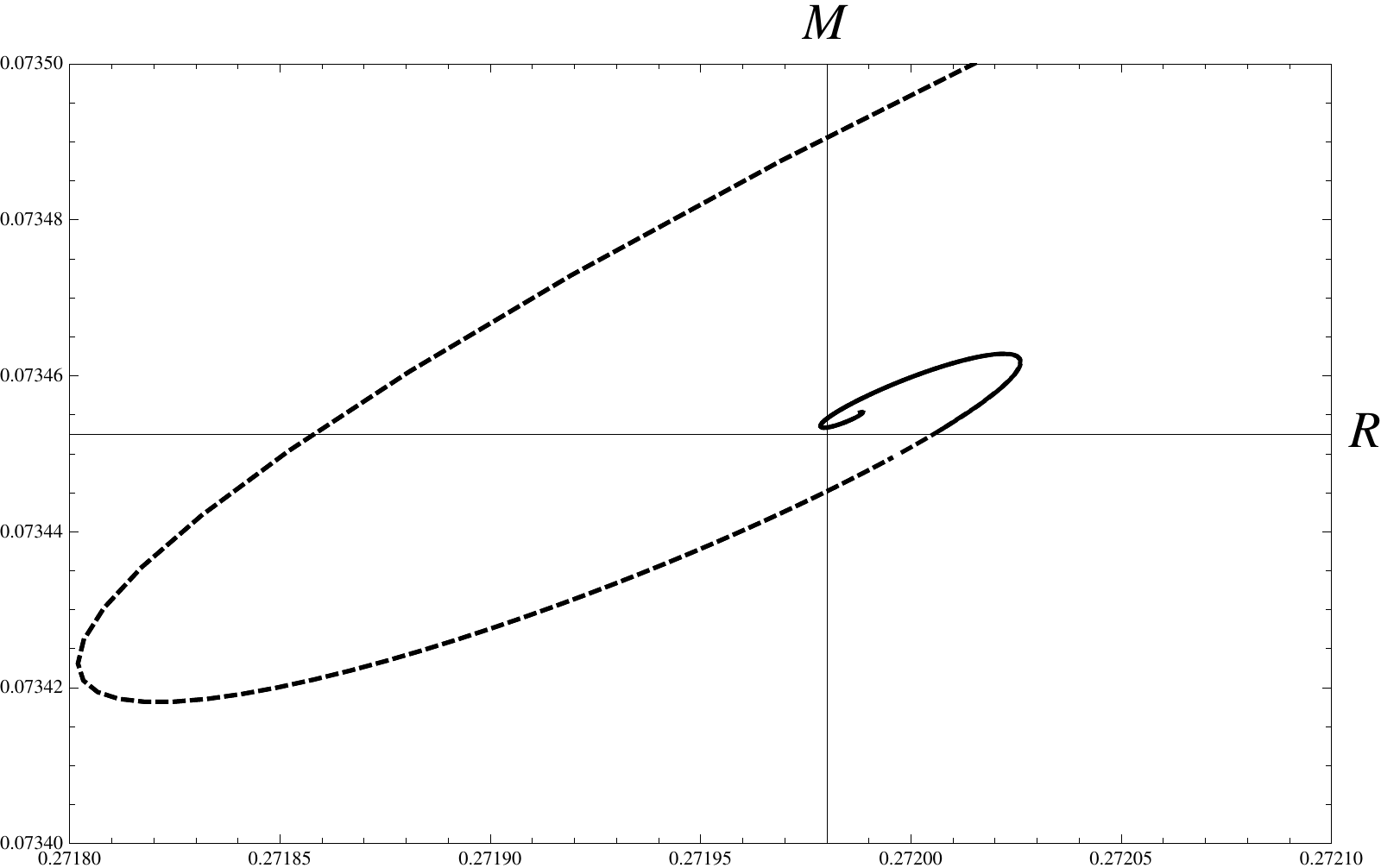}}
\caption{Mass-radium diagram for Vlasov matter (left) and elastic matter (right). The picture for elastic matter is a magnification of the spiral $k=0.05$ shown in Figure~\ref{RM}.}\label{spiralsfig}
\end{figure}

Finally we comment on the difference between the rigorous results on static solutions established for the two matter models. As opposed to elastic matter the existence theory for spherically symmetric static solutions of the Einstein-Vlasov system is well developed. Existence of solutions with compact support and finite mass have been obtained, and a sub-class of the solutions admit the quantity $2M/R$ to be arbitrarily close to the maximum value $8/9$. Hence, these solutions are in particular not small data solutions. We refer to~\cite{H} for a review.

\appendix
\section*{Appendix}
This Appendix contains a quick introduction to elasticity theory in general relativity. We restrict ourselves to derive the system~\eqref{h1h2m} and the constitutive equations~\eqref{consteq} used in this paper. For more background on the subject we refer to~\cite{BS, CQ, KS, KM}.

The {\it reference state} of the elastic body is a three dimensional Riemannian manifold $(N,\gamma)$ which identifies the body in the undeformed state (no strain).
The deformed state of the body in a four dimensional space-time $(M,g)$ is described by the
\textit{configuration map}
$\psi:M\to N$. It is assumed that for all $q\in N$ the set $\psi^{-1}(q)$ is a timelike curve (the world-line of the ``particle" $q$). This definition implies that the kernel of the \textit{deformation gradient}
$d\psi: TM \rightarrow TN$ is generated by a (future-directed unit) timelike
vector field $u$, which is the four-velocity field of the body particles.

Given a system of coordinates $x^\alpha$  in space-time ($\alpha=0,\dots, 3$) and a system of coordinates
$X^A$ in the reference state ($A=1,2,3$), we denote by $h_{\alpha\beta}$ the components of the pull-back of the material metric $\gamma$:
\begin{equation}\label{hdef}
h_{\alpha\beta}=\partial_\alpha\psi^A\partial_\beta\psi^B\,\gamma_{AB},
\end{equation}
which is also called the {\it relativistic strain tensor}. Clearly $h_{\alpha\beta}u^\beta=0$,  hence $h^\alpha_{\ \beta}=g^{\alpha\sigma}h_{\sigma\beta}$ has three positive
eigenvalues $h_1$, $h_2$, $h_3$, which are called the {\it principal stretches} of the body.
 The scalar quantity
\[
n=\sqrt{h_1h_2h_3}
\]
is the \textit{particle density} of the material. A {\it constitutive equation} for the elastic material is a function $w: (0,\infty)^3\to [0,\infty)$ such that the Lagrangian density $\mathcal{L}$ of the matter takes the form
\begin{equation}\label{lagrangian}
\mathcal{L}=w(h_1,h_2,h_3).
\end{equation}
A constitutive equation is said to have a  {\it quasi-Hookean} form if
\begin{subequations}\label{john}
\begin{equation}\label{stored}
w(h_1,h_2,h_3)=\hat{\rho}_0(n)+\hat{\mu}(n) F(\sigma):=\hat{w}(n,\sigma),
\end{equation}
where $\sigma$ is a {\it shear scalar}, i.e., a  non-negative function of the principal stretches such that $\sigma=0$ iff $h_1=h_2=h_3$, $F$ is a non-negative smooth function such that $F(0)=0$, $\hat{\rho}_0$ is the {\it unsheared energy density} and $\hat{\mu}(n)$ the {\it rigidity modulus} of the elastic material. In this paper we adopt the shear scalar used in~\cite{T}, namely
\begin{equation}\label{shear}
\sigma=\frac{h_1+h_2+h_3}{n^{2/3}}-3,
\end{equation}
and moreover we choose
\begin{equation}
F(\sigma)=\sigma.
\end{equation}
\end{subequations}
Following~\cite{T}, we call these elastic materials {\it John materials}.

The stress-energy tensor of elastic materials is obtained as the variation with respect to the
spacetime metric of the matter action $S_M=-\int\sqrt{|g |}\,\mathcal{L}$. For the constitutive equation~\eqref{lagrangian}-\eqref{john} we obtain
 \begin{equation}\label{moregeneralT}
T_{\alpha\beta}=\rho\, u_\alpha u_\beta+p\,(g_{\alpha\beta}+u_\alpha u_\beta)+\frac{2\mu}{n^{2/3}}\left(h_{\alpha\beta}-\frac{\sum_{i=1}^3h_i}{3}(g_{\alpha\beta}+u_\alpha u_\beta)\right),
\end{equation}
where $\mu(x^\alpha)=\hat{\mu}(n(x^\alpha))$,  $\rho(x^\alpha)=\hat{w}(n(x^\alpha),\sigma(x^\alpha))$, $p(x^\alpha)=\hat{p}(n(x^\alpha),\sigma(x^\alpha))$ and
\[
\hat{p}(n,\sigma )=\hat{p}_0(n)+\hat{\nu}(n)\sigma,\quad \hat{p}_0(n)=n^2\frac{d}{dn}\left(\frac{\hat{\rho}_0(n)}{n}\right),\quad \hat{\nu}(n)=n^2\frac{d}{dn}\left(\frac{\hat{\mu}(n)}{n}\right).
\]
Now we assume
\begin{equation}\label{unsheared}
\hat{\mu}(n)=k\hat{\rho}_0(n),
\end{equation}
for some constant $k\geq 0$, called the {\it rigidity parameter} of the elastic body.
This implies that
\begin{equation}\label{isotropicpressure}
\hat{w}(n,\sigma)=\hat{\rho}_0(n)(1+k \sigma).
\end{equation}
Moreover, denoting by $u,v_{(1)}, v_{(2)}, v_{(3)}$ the orthonormal eigenvectors of $h^\alpha_{\ \beta}$ corresponding to the eigenvalues $(0,h_1,h_2,h_3)$ (by definition, $u$ is the eigenvector of the zero eigenvalue), we have
\[
T^\alpha_{\ \beta}u^\beta=-\rho u^\alpha,\qquad T^\alpha_{\ \beta}v_{(i)}^\beta=p_iv_{(i)}^\alpha,
\]
where the anisotropic pressures are $p_i(x^\alpha)=\hat{p}_i(h_1(x^\alpha),h_2(x^\alpha),h_3(x^\alpha))$,
\begin{equation}\label{anipress}
\hat{p}_i(h_1,h_2,h_3)=\hat{p}_0(n)(1+k\sigma)+2k\frac{\hat{\rho}_0(n)}{n^{2/3}}\left(h_{i}-\frac{1}{3}\sum_{i=1}^3h_i\right).
\end{equation}

{\bf Static spherically symmetric elastic bodies}. Assume now that spacetime $(M,g)$ is static and spherically symmetric. This entails that the body is also static and spherically symmetric, which means that the static Killing vector field coincides with the four-velocity $u$ and that $h$ is invariant by the rotational group of isometries, see~\cite{BCV}. It is clear that in this case there are only two independent eigenvalues of the relativistic strain tensor: the {\it radial stretch} $x=h_1$ and the {\it tangential stretch} $y=h_2=h_3$. Moreover $x$ and $y$ are functions of the radial variable only; we fix a system $(t,r,\theta,\phi)$ of Schwarzschild coordinates in spacetime and express the metric in the form
\[
g=-e^{2\beta(r)}\,dt^2+e^{2\lambda(r)}\,dr^2+r^2(d\theta^2+\sin^2\theta d\phi^2).
\]
It follows that
\[
n=\sqrt{x}y,\quad \sigma=\frac{x+2y}{n^{2/3}}-3
\]
and so, by~\eqref{isotropicpressure},~\eqref{anipress} the matter fields constitutive equations take the form
\begin{align*}
&\hat{w}(n,\sigma)=\hat{\rho}(x,y), \quad \hat{p}_1(h_1,h_2,h_3)=\hat{p}_\mathrm{rad}(x,y),\\
 &\hat{p}_2(h_1,h_2,h_3)=\hat{p}_3(h_1,h_2,h_3)=\hat{p}_\mathrm{tan}(x,y)
\end{align*}
where $\hat{\rho}$, $p_\mathrm{rad}$, $\hat{p}_\mathrm{tan}$ are given by~\eqref{consteq}.
Finally we derive the system~\eqref{h1h2m}. First we have
\[
p_\mathrm{rad}'=\partial_x\hat{p}_\mathrm{rad}(x,y)x'+\partial_y\hat{p}_\mathrm{rad}(x,y)y',
\]
hence~\eqref{eqh1} follows from the TOV equation~\eqref{TOVeq}. Equation~\eqref{meq} follows by the definition of Hawking mass. It remains to establish~\eqref{eqh2}.  The latter equation depends on the choice of the material metric $\gamma$  in the reference state. We choose $\gamma$ to be the Euclidean metric, which we express in spherical coordinates $(R,\Theta,\Phi)$ as
\[
\gamma=dR^2+R^2(d\Theta^2+\sin^2\Theta d\Phi^2).
\]
In these coordinates the (spherically symmetric) configuration map becomes a function $(R(r),\Theta(\theta,\phi),\Phi(\theta,\phi))$. By rotating the axes we can assume $\Theta=\theta$, $\Phi=\phi$, hence
\[
h=(R'(r))^2dr^2+R^2(r)(d\theta^2+\sin^2\theta d\phi^2).
\]
It follows that
\[
h_1=x=e^{-2\lambda}(R')^2,\quad y=\frac{R^2}{r^2}.
\]
Differentiating $y(r)$ and using that $e^{-2\lambda}=1-2m/r$ yields~\eqref{eqh2}.

\end{document}